\documentclass[11pt,reqno]{amsart}

\usepackage{amssymb}
\usepackage{amscd}
\usepackage{amsfonts}
\usepackage{mathrsfs}
\usepackage{setspace}
\usepackage{version}
\usepackage{mathtools}
\usepackage{enumerate}
\usepackage[english]{babel}
\usepackage{xcolor}

\newtheorem{theorem}{Theorem}[section]
\newtheorem{lemma}[theorem]{Lemma}

\newtheorem{corollary}[theorem]{Corollary}
\theoremstyle{definition}
\newtheorem{definition}[theorem]{Definition}

\theoremstyle{remark}
\newtheorem{remark}[theorem]{Remark}
\numberwithin{equation}{section}

\newcommand{\field}[1]{\mathbb{#1}}

\newcommand{\R}{\field{R}}
\newcommand{\N}{\field{N}}

\newcommand{\al}{\alpha}

\newcommand{\de}{\delta}
\newcommand{\ep}{\varepsilon}

\newcommand{\la}{\lambda}

\newcommand{\ph}{\varphi}

\newcommand{\ps}{\psi}
\newcommand{\om}{\omega}

\newcommand{\Om}{\Omega}

\newcommand{\dist}{\mathop{\text{\upshape{dist}}}\nolimits}

\newcommand{\ca}{\mathop{\text{\upshape{ca}}}\nolimits}

\newcommand\mn[1]{#1}

\begin{document}

	\title[Continuity from below of monotone functionals on $C_b$]{Lower semicontinuity of monotone functionals in the mixed topology on $C_b$}
	
	\author{Max Nendel}
	\address{Center for Mathematical Economics, Bielefeld University}
	\email{max.nendel@uni-bielefeld.de}

	\thanks{The author thanks Ben Goldys, Michael Kupper, Markus Kunze, as well as two anonymous referees for helpful comments related to this work.\ Financial support through the Deutsche Forschungsgemeinschaft (DFG, German Research Foundation) -- SFB 1283/2 2021 -- 317210226 and the Australian Research Council Discovery Project DP120101886 is gratefully acknowledged.}
	
	\date{\today}

	\begin{abstract}
	The main result of this paper characterizes the continuity from below of monotone functionals on the space $C_b$ of bounded continuous functions on an arbitrary Polish space as lower semicontinuity in the mixed topology.\ In this particular situation, the mixed topology coincides with the Mackey topology for the dual pair $(C_b,\ca)$, where $\ca$ denotes the space of all countably additive signed Borel measures of finite variation.\ Hence, lower semicontinuity in the mixed topology of convex monotone maps $C_b\to \R$ is equivalent to a dual representation in terms of countably additive measures.\ Such representations are of fundamental importance in finance, e.g., in the context of risk measures and super hedging problems.\ \mn{Based on the main result, regularity properties of capacities and dual representations of Choquet integrals in terms of countably additive measures for $2$-alternating capacities are studied.\ In a second step,} the paper provides a characterization of equicontinuity in the mixed topology for families of convex monotone maps.\ As a consequence, for every convex monotone map on $C_b$ taking values in a locally convex vector lattice, continuity in the mixed topology is equivalent to continuity on norm bounded sets.\smallskip
	\end{abstract}

	\maketitle
	
	\setcounter{tocdepth}{1}

\vspace{1em} 
\section{Introduction}

In this paper, we study continuity properties for monotone maps $C_b\to \R$, where $C_b=C_b(\Om)$ denotes the space of all bounded continuous functions on a Polish space $\Om$ with values in $\R$. Monotone functionals $C_b\to \R$ appear in many applications.\ Special instances of such maps, in the context of \mn{finance and actuarial science}, are
\begin{itemize}
	\item risk measures or nonlinear expectations, cf.\ \cite{MR2754968,MR3859905,MR3970247},
	\item super hedging functionals, cf.\ \cite{MR4238277,MR4289666},
	\item robust expected utilities or loss functions, cf.\ \cite{delbaen1,delbaen2},
	\item Choquet integrals, e.g., in the context of insurance premia, cf.\ \cite{MR1604936}.
\end{itemize}	
In order to obtain dual representations of convex monotone functionals in terms of countably additive measures, additional continuity properties are usually required. The two most prominent continuity properties in this context are continuity from above and continuity from below, cf.\ \cite{MR3859905}. For convex monotone functionals, continuity from below is usually a weaker requirement than continuity from above, see, for instance, \cite{MR4289666}.\

\mn{In the field of mathematical finance, these continuity properties have been studied in many contexts in the past decades.\ For risk measures, continuity from above is (up to a different sign convention) related to the Lebesgue property, whereas continuity from below is closely tied to the Fatou property.\ The Fatou property or, more generally, Fatou closedness is a fundamental ingredient in no-arbitrage theory, see, e.g., \cite{MR4135699} and \cite{MR4370806}.\ Fixing a reference measure and working on $L^\infty$, it is well-known that continuity from below of convex risk measures is equivalent to a dual representation in terms of countably additive measures, cf.\ \cite[Theorem 4.33]{MR3859905} and, if the risk measure is also law-invariant, continuity from below is automatically satisfied if the underlying probability space is assumed to be atomless, cf.\ \cite{MR2277714}.

On the other hand, monetary risk measures are closely linked to nonlinear expectations and the topic of model uncertainty in finance. In this context, risk measures, which are not dominated by a single probability measure that deems events to be negligible or not, play a crucial role.\ An example for such a risk measure is the $G$-expectation, cf.\ \cite{MR2754968} and \cite{MR3970247}.\ Working, however, on the space $B_b$ of bounded measurable functions without fixing a reference probability, continuity from below alone is \textit{not} sufficient in order to guarantee a dual representation of convex monotone functionals in terms of countably additive measures despite the fact that it implies sequential lower semicontinuity of such functionals in the weak topology $\sigma(B_b,\ca)$ of the dual pair $(B_b,\ca)$.\ In \cite[Example 3.6]{denk2018kolmogorov}, an example for a coherent risk measure, which is continuous from below on $B_b$ but does not have a single countably additive minorant, is given. On the other hand, continuity from above of a risk measure on the space of bounded measurable functions already implies the existence of a dominating reference measure, see, e.g., \cite[Remark 3.3]{denk2018kolmogorov}.

One way out of this dilemma is to restrict the attention to continuous claims.\ Choosing the space $C_b$ of bounded continuous functions as an underlying function space, it is well-known that continuity from above, for a convex monotone functional $C_b\to \R$, is a sufficient but not necessary condition for a dual representation in terms of countably additive measures, see, e.g., \cite{MR4289666}.\ However, the question whether such a representation is equivalent to the weaker notion of continuity from below on general Polish state spaces has been unanswered for almost a decade, as discussed in the introduction of \cite{delbaen1}. In a series of papers \cite{delbaen1,delbaen2}, this question has been answered positively and, as a consequence, convex monotone functionals on $C_b$, which are continuous from below, are lower semicontinuous in the Mackey topology $\mu(C_b,\ca)$ of the dual pair $(C_b,\ca)$, where $\ca$ denotes the space of all countably additive signed Borel measures with finite variation.\ From a mathematical perspective, this is a remarkable result, since the Mackey topology $\mu(C_b,\ca)$ is not metrizable and continuity from below is a requirement for sequences, so that nonmetrizability poses a major problem.\ Therefore, in \cite{delbaen1,delbaen2}, another path is chosen, and the proofs therein rely on compactification methods, more precisely, on the fact that every Polish space can be embedded as a $G_\delta$ into a compact metric space.}

The present paper generalizes the main result of \cite{delbaen2} and proves that, for \textit{any} monotone functional $C_b\to \R$, continuity from below is equivalent to lower semicontinuity in the Mackey topology $\mu(C_b,\ca)$, which is also known as the mixed topology in this setting, cf.\ Theorem \ref{thm.main}.\ Moreover, Theorem \ref{thm.main} shows that, for monotone functionals $C_b\to \R$, upper or lower semicontinuity in the mixed topology are equivalent to sequential upper or lower semicontinuity in the mixed topology, respectively, despite the fact that the mixed topology is not metrizable.\ \mn{Since the Mackey topology $\mu(C_b,\ca)$ is the finest topology leading to the dual space $\ca$ of countably additive signed Borel measures with finite variation, it is a natural choice for duality theory on $C_b$. In particular, Theorem \ref{thm.main} implies that every convex monotone functional on $C_b$ admits a dual representation in terms of countably additive measures, cf.\ Corollary \ref{cor.convex}.\ Using, however, the explicit representation of a local base of the origin of the mixed topology allows to further characterize continuity from below in terms of proximity on compact sets, see Theorem \ref{thm.main}.}

\mn{In Corollary \ref{cor.choquet}, we turn our focus on capacities and Choquet integrals. In a first step, we characterize the regularity of general capacities, defined on open sets, in terms of continuity from below of the Choquet integral on the set $L_b$ of all bounded lower semicontinuous functions $\Om\to \R$ and the capacity along sequences of open sets. In a second step, we characterize $2$-alternating capacities, for which the related Choquet integral admits a dual representation in terms of countably additive measures on $L_b$, in terms of continuity from below along sequences of open sets and regularity of the capacity, partially extending the results from \cite{MR0444879}.}

Another question that we address in this paper is a characterization of continuity of convex monotone maps in the mixed topology.\ Continuity in the mixed topology is of fundamental importance in many situations in robust finance.\ In the context of super hedging, it has been studied in \cite{MR4238277}.\ In the context of dynamic risk measures and semigroups related to stochastic processes under model uncertainty, it has appeared in \cite{nend2,nend1} and, with a different language, it is also used in the analysis of large deviations principles based on max-stable risk measures, cf.\ \cite{MR4255224}.\ Building on Theorem \ref{thm.main}, we discuss the equicontinuity of families of convex monotone maps in the mixed topology in Theorem \ref{thm.supremum}.\ There, a characterization of equicontinuity in terms of continuity from above and uniform equicontinuity on supremum norm bounded sets is given.

\mn{From a mathematical perspective, the mixed topology has two striking features.\ On the one hand, unless $\Om$ is compact, it has \textit{no} neighborhood of zero, which is bounded with respect to the supremum norm.\ On the other hand, it has the intrinsic property that, for linear operators taking values in an arbitrary locally convex space, continuity is equivalent to continuity on supremum norm bounded sets.\ Corollary \ref{cor.lcs}, which is a consequence of Theorem \ref{thm.supremum}, extends this intrinsic property showing that, for convex monotone maps taking values in a locally convex vector lattice, $\mu(C_b,\ca)$-continuity is equivalent to $\mu(C_b,\ca)$-continuity on supremum norm bounded sets.\ A particular instance of a locally convex vector lattice is $C_b$ itself, endowed with the mixed topology, which, in a financial context, corresponds, for example, to the case of conditional risk measures or conditional nonlinear expectations.}

The rest of the paper is organized as follows.\ In Section \ref{sec.main}, we state the main results and their corollaries.\ Section \ref{sec.proofs} contains all proofs.\ In the Appendix \ref{app.capacity}, we prove an auxiliary result for capacities and Choquet integrals and, in the Appendix \ref{app.lcs}, we prove an auxiliary result on locally convex vector lattices.  

\section{Main results}\label{sec.main}
 Throughout, let $\Om$ be a Polish space and $C_b=C_b(\Om)$ denote the space of all bounded continuous functions $\Om\to \R$. We consider the local base $$\mathcal V^2:=\big\{\{g\in C_b\, |\, \|g\|_\infty<r\}\, \big|\, r>0\big\}$$ of $0\in C_b$ for the topology induced by the supremum norm $\|\cdot\|_\infty$ and the local base
 $$
 \mathcal V^1:=\bigg\{\Big\{g\in C_b\, \Big|\, \sup_{x\in C}|g(x)|<r\Big\}\, \bigg|\, r>0, \, C\subset \Omega\text{ compact}\bigg\}
 $$
  of $0\in C_b$ for the vector topology of uniform convergence on compacts. Let $\mathcal V$ denote the system  consisting of all sets
  \[
  \bigcup_{n\in \N}\sum_{k=1}^n (V_k^1\cap kV^2)
  \]
 with $(V_k^1)_{k\in \N}\subset \mathcal V^1$ and $V^2\in \mathcal V^2$, where $kV^2:=\{kg\, |\, g\in V^2\}$ for all $k\in \N$ and $$\sum_{k=1}^n V_k:=\bigg\{\sum_{k=1}^ng_k\, \bigg|\, g_1\in V_1,\ldots, g_n\in V_n\bigg\}$$ for nonempty subsets $V_1,\ldots, V_n$ of $C_b$ and $n\in \N$.
 
  Then, $\mathcal V$ is a local base of $0\in C_b$ for a Hausdorff locally convex topology $\beta$, which is known as the \textit{mixed topology}.\ We refer to \cite{wiweger} for a detailed discussion on the mixed topology in a more general setting.\ Clearly, the mixed topology $\beta$ is finer than the \textit{weak topology} $\sigma(C_b,\ca)$ of the dual pair $(C_b,\ca)$, where $\ca$ denotes the space of all countably additive signed Borel measures of finite variation.\ A well-known fact, which we will \textit{not} make use of, is that the mixed topology $\beta$ coincides with the Mackey topology of the dual pair $(C_b,\ca)$.\ Moreover, $\beta$ belongs to the class of \textit{strict topologies}, cf.\ \cite{wheeler1}.\ We also refer to \cite{haydon} and \cite{sentilles} for additional fine properties of mixed or strict topologies.
  
  We say that a functional $U\colon C_b\to \R$ is \textit{monotone} if $U(f)\leq U(g)$ for all $f,g\in C_b$ with $f\leq g$, where, for functions $\Omega\to \R$, the relation $\leq$ and all other order-related objects refer to the pointwise order.
  
  For a sequence $(f_n)_{n\in \N}\subset C_b$ and $f\colon \Om\to \R$, we write $f_n\nearrow f$ as $n\to \infty$ if $f_n\leq f_{n+1}$ for all $n\in \N$ and $f(x)=\lim_{n\to \infty} f_n(x)$ for all $x\in \Om$.\ Analogously, we write  $f_n\searrow f$ as $n\to \infty$ if $f_n\geq f_{n+1}$ for all $n\in \N$ and $f(x)=\lim_{n\to \infty} f_n(x)$ for $(f_n)_{n\in \N}\subset C_b$ and $f\in C_b$.
  
  \begin{definition}\
  	\begin{enumerate}
    \item[a)] We say that a monotone functional $U\colon C_b\to \R$ is \textit{continuous from below} if $U(f)=\lim_{n\to \infty}U(f_n)$ for all sequences $(f_n)_{n\in \N}\subset C_b$ and $f\in C_b$ with $f_n\nearrow f$ as $n\to \infty$.
    \item[b)] We say that a monotone functional $U\colon C_b\to \R$ is \textit{continuous from above} if $U(f)=\lim_{n\to \infty}U(f_n)$ for all sequences $(f_n)_{n\in \N}\subset C_b$ and $f\in C_b$ with $f_n\searrow f$ as $n\to \infty$.
  \end{enumerate}
\end{definition}

\begin{theorem}\label{thm.main}
	Let $U\colon C_b\to \R$ be monotone.\ Then, the following statements are equivalent.
	\begin{enumerate}
		\item[(i)] $U$ is continuous from below.
		\item[(ii)] $U$ is lower semicontinuous in the mixed topology $\beta$.
		\item[(iii)] $U$ is sequentially lower semicontinuous in the mixed topology $\beta$.
        \mn{\item[(iv)] For all $f\in C_b$, $\ep>0$, and $r\geq 0$, there exist $\de>0$ and a compact $C\subset \Omega$ such that
        \[
        U(f)\leq U(f+e)+\ep
        \]
        for all $e\in C_b$ with $\|e\|_\infty\leq r$ and $\sup_{x\in C}|e(x)|\leq \de$.}
	\end{enumerate}
\end{theorem}

Considering, for an arbitrary monotone function $U\colon C_b\to \R$, its conjugate function $\overline U\colon C_b\to \R$, given by $\overline U(f):=-U(-f)$ for all $f\in C_b$, we obtain the following corollary.
\begin{corollary}\label{cor.main}
	Let $U\colon C_b\to \R$ be monotone.\ Then, the following statements are equivalent.
	\begin{enumerate}
		\item[(i)] $U$ is continuous from above.
		\item[(ii)] $U$ is upper semicontinuous in the mixed topology $\beta$.
		\item[(iii)] $U$ is sequentially upper semicontinuous in the mixed topology $\beta$.
        \mn{\item[(iv)] For all $f\in C_b$, $\ep>0$, and $r\geq 0$, there exist $\de>0$ and a compact $C\subset \Omega$ such that
        \[
        U(f+e)\leq U(f)+\ep
        \]
        for all $e\in C_b$ with $\|e\|_\infty\leq r$ and $\sup_{x\in C}|e(x)|\leq \de$.}
	\end{enumerate}
\end{corollary}
A combination of Theorem \ref{thm.main} and Corollary \ref{cor.main} leads to the following characterization of continuity in the mixed topology for monotone functionals.
\begin{corollary}\label{cor.main.cont}
	Let $U\colon C_b\to \R$ be monotone.\ Then, the following statements are equivalent.
	\begin{enumerate}
		\item[(i)] $U$ is continuous from above and below.
		\item[(ii)] $U$ is continuous in the mixed topology $\beta$.
		\item[(iii)] $U$ is sequentially continuous in the mixed topology $\beta$.
        \mn{\item[(iv)] For all $f\in C_b$, $\ep>0$, and $r\geq 0$, there exist $\de>0$ and a compact $C\subset \Omega$ such that
        \[
        |U(f+e)-U(f)|\leq \ep
        \]
        for all $e\in C_b$ with $\|e\|_\infty\leq r$ and $\sup_{x\in C}|e(x)|\leq \de$.}
	\end{enumerate}
\end{corollary}

As a direct consequence of Theorem \ref{thm.main}, we obtain the main result in \cite{delbaen2}.\ We denote by $\ca_+$ the set of all positive elements of $\ca$, i.e., the set of all finite Borel measures.

\begin{corollary}\label{cor.convex}
	Let $U\colon C_b\to \R$ be convex and monotone.\ Then, the following are equivalent.
	\begin{enumerate}
		\item[(i)] $U$ is continuous from below.
		\item[(ii)] $U$ is lower semicontinuous in the mixed topology $\beta$.
		\item[(iii)] $U$ is lower semicontinuous in the weak topology $\sigma(C_b,\ca)$.
		\item[(iv)] There exists a nonempty set $\mathcal M\subset \ca_+$ and a function $\alpha\colon \mathcal M\to \R$ such that
		\begin{equation}\label{eq.dual}
			U(f)=\sup_{\mu\in \mathcal M} \bigg(\int_\Omega f\, {\rm d}\mu-\alpha(\mu)\bigg)\quad \text{for all }f\in C_b.
		\end{equation}
	\end{enumerate}	
\end{corollary}	

We apply Corollary \ref{cor.convex} to the case of Choquet integrals. In the sequel, let $\mathcal O$ denote the set of all open subsets of $\Omega$, i.e., the topology on $\Omega$, and $L_b$ the set of all bounded lower semicontinuous functions $\Om\to \R$.\ A \textit{capacity (on $\mathcal O$)} is a map $c\colon \mathcal O\to [0,\infty)$ with 
$$c(\emptyset)=0\quad\text{and}\quad c(B_1)\leq c(B_2)\quad\text{for all }B_1,B_2\in \mathcal O\text{ with }B_1\subset B_2.$$
For a capacity $c\colon \mathcal O\to [0,\infty)$, we define the \textit{Choquet integral} with respect to $c$ as
\[
\int_\Om f\, {\rm d}c:=\int_0^{\infty} c(f>s)\, {\rm d}s+\int_{-\infty}^0 \big(c(f>s)-c(\Om)\big)\, {\rm d}s\quad \text{for all }f\in L_b.
\]
By definition, the Choquet integral is \textit{positively homogeneous}, i.e., $\int_\Om (\lambda f)\,{\rm d}c=\lambda\int_\Om f\,{\rm d}c$ for all $f\in L_b$ and $\lambda>0$, and \textit{constant additive}, i.e., $\int_\Om ( f+m)\,{\rm d}c=\int_\Om f\,{\rm d}c+mc(\Om)$ for all $f\in L_b$ and $m\in \R$. A well-known fact is that the Choquet integral is \textit{subadditive}, i.e.,  $\int_\Om (f_1+f_2)\,{\rm d}c\leq \int_\Om f_1\,{\rm d}c+\int_\Om f_2\,{\rm d}c$ for all $f_1,f_2\in L_b$, if and only if the capacity $c$ is \textit{$2$-alternating}, i.e.,
\[
c(B_1\cup B_2)+c(B_1\cap B_2)\leq c(B_1)+c(B_2)\quad\text{for all }B_1,B_2\in \mathcal O.
\]
For the reader's convenience, we provide a proof of this statement in the Appendix \ref{app.capacity}.

Another consequence of Theorem \ref{thm.main} and Corollary \ref{cor.convex} is the following result concerning the regularity of general capacities and dual representations of Choquet integrals in terms of countably additive measures for $2$-alternating capacities.
\begin{corollary}\label{cor.choquet}
	Let $c\colon \mathcal O\to [0,\infty)$ be a capacity.\ Then, the following statements are equivalent.
	\begin{enumerate}
		\item[(i)] For every sequence $(B_n)_{n\in \N}\subset \mathcal O$ with $B_n\subset B_{n+1}$ for all $n\in \N$,
		\[
		 c\bigg(\bigcup_{n\in \N}B_n\bigg)=\lim_{n\to\infty}c(B_n).
		\]
        \item[(ii)] The Choquet integral is continuous from below on $L_b$, i.e., 
        \[
            \int_\Omega f\,{\rm d} c=\lim_{n\to \infty} \int_\Omega f_n\,{\rm d} c
        \]
        for all $f\in L_b$ and any sequence $(f_n)_{n\in \N}\subset L_b$ with $f_n\nearrow f$ as $n\to \infty$.
        \item[(iii)] The capacity $c$ is regular, i.e., for all $B\in \mathcal O$,
        \begin{equation}\label{eq.regularity}
        c(B)=\sup_{C\Subset B}\inf_{\substack{A\in \mathcal O\\ A\supset C}}c(A).\footnote{\text{\mn{Here and throughout, we write $C\Subset B$ if $C\subset B\subset \Omega$ and $C$ is compact.}}}
        \end{equation}
        \end{enumerate}	
        If $c$ is $2$-alternating, the statements (i) - (iii) are equivalent to the following statement:
        \begin{enumerate}
		\item[(iv)] There exists a nonempty set $\mathcal M\subset \ca_+$ with $\mu(\Om)=c(\Om)$ for all $\mu\in \mathcal M$ and
		\[
		\int_\Omega f\, {\rm d}c=\sup_{\mu\in \mathcal M}\int_\Omega f\, {\rm d}\mu\quad\text{for all }f\in L_b.
		\]
	\end{enumerate}	
\end{corollary}

We conclude this section with various characterizations of continuity in the mixed topology for convex monotone maps. We start with the following theorem, which is the second main result.

\begin{theorem}\label{thm.supremum}
  Let $I$ be a nonempty index set and $(U_i)_{i\in I}$ be a family of convex and monotone maps $C_b\to \R$ with
   \begin{equation}\label{eq.ass.thm.supremum}
   \sup_{i\in I} \big(U_i(r)-\mn{U_i(0)}\big)<\infty\quad \text{for all constants }r\geq 0.
   \end{equation}
Then, the following statements are equivalent.
	\begin{enumerate}
		\item[(i)] For every sequence $(f_n)_{n\in \N}\subset C_b$ with $f_n\searrow 0$ as $n\to \infty$,
		\[
		\lim_{n\to \infty}\sup_{i\in I} \big(U_i(f_n)- U_i(0)\big)=0.
		\]
		\item[(ii)] For every $r\geq 0$ and every $\ep >0$, there exists some $V\in \mathcal V$ with 
		\[
		\sup_{\|f\|_\infty\leq r}\sup_{i\in I} |U_i(f+e)-U_i(f)|<\ep \quad \text{for all }e\in V.
		\]
		\item[(iii)] For every $r\geq 0$ and every $\ep >0$, there exists a compact $C\subset \Omega$ and a constant $M\geq 0$ such that
		\[
		\sup_{i\in I}|U_i(f_1)-U_i(f_2)|\leq M\sup_{x\in C}|f_1(x)-f_2(x)|+\ep
		\] 
		for all $f_1,f_2\in C_b$ with $\max\big\{\|f_1\|_\infty,\|f_2\|_\infty\big\}\leq r$.
	\end{enumerate}	
\end{theorem}

Theorem \ref{thm.supremum} leads to the following characterization for continuity of convex and monotone maps on $C_b$ taking values in a \textit{locally convex vector lattice} $(L,\tau)$, i.e., a vector lattice $L$ together with a locally convex topology $\tau$ on $L$, which is generated by a family of lattice seminorms.\ Recall that, for a vector lattice $L$, a seminorm $p\colon L\to [0,\infty)$ is called a \textit{lattice seminorm} if $p(u)\leq p(v)$ for all $u,v\in L$ with $|u|\leq |v|$.\ We refer to \cite{MR1741419} for a detailed study of locally convex vector lattices. For the reader's convenience, we provide an auxiliary result on locally convex vector lattices in the Appendix \ref{app.lcs}.

\begin{corollary}\label{thm.lcs}
	Let $(L,\tau)$ be a locally convex vector lattice and $U\colon C_b\to L$ be convex and monotone.\ Then, the following statements are equivalent.
	\begin{enumerate}
		\item[(i)] $U$ is $\beta$-$\tau$-continuous.
		\item[(ii)] For every nonnegative $\tau$-continuous linear functional $\lambda\colon L\to \R$ and every sequence $(f_n)_{n \in \N}\subset C_b$ with $f_n\searrow 0$ as $n\to \infty$,
		\[
		\lim_{n\to \infty}\lambda \big(U(f_n)\big)=\lambda\big(U(0)\big).
		\]
		\item[(iii)] For every nonnegative $\tau$-continuous linear functional $\lambda\colon L\to \R$, the map
		\[
		 C_b\to \R, \; f\mapsto \lambda\big(U(f)\big)
		\]
		 is sequentially upper semicontinuous in the mixed topology $\beta$.
		 \item[(iv)] For every $\tau$-continuous lattice seminorm $p\colon L\to [0,\infty)$, every constant $r\geq 0$, and all $\ep >0$, there exists some $V\in \mathcal V$ with 
		 \[
		 \sup_{\|f\|_\infty\leq r} p\big(U_i(f+e)-U_i(f)\big)<\ep \quad \text{for all }e\in V.
		 \]
		\item[(v)] For every $\tau$-continuous lattice seminorm $p\colon L\to [0,\infty)$, every constant $r\geq 0$, and all $\ep >0$, there exists a compact $C\subset \Omega$ and a constant $M\geq 0$ such that
		\[
		p\big(U(f_1)-U(f_2)\big)\leq M\sup_{x\in C}|f_1(x)-f_2(x)|+\ep
		\] 
		for all $f_1,f_2\in C_b$ with $\max\big\{\|f_1\|_\infty,\|f_2\|_\infty\big\}\leq r$.
	\end{enumerate}	
\end{corollary}

Choosing $L=\R$ with the topology induced by the absolute value $|\cdot|$, we obtain the following corollary.
\begin{corollary}\label{cor.cont.convex}
	Let  $U\colon C_b\to \R$ be convex and monotone.\ Then, the following statements are equivalent.
	\begin{enumerate}
		\item[(i)] $U$ is continuous from above.
		\item[(ii)] $U$ is sequentially upper semicontinuous in the mixed topology $\beta$.
		\item[(iii)] $U$ is continuous in the mixed topology $\beta$.
		\item[(iv)] For every $r\geq 0$ and every $\ep >0$, there exists some $V\in \mathcal V$ with 
		\[
		\sup_{\|f\|_\infty\leq r} |U(f+e)-U(f)|<\ep \quad \text{for all }e\in V.
		\]
		\item[(v)] For every $r\geq 0$ and every $\ep >0$, there exists a compact $C\subset \Omega$ and a constant $M\geq 0$ such that
		\[
		|U(f_1)-U(f_2)|\leq M\sup_{x\in C}|f_1(x)-f_2(x)|+\ep
		\] 
		for all $f_1,f_2\in C_b$ with $\max\big\{\|f_1\|_\infty,\|f_2\|_\infty\big\}\leq r$.
	\end{enumerate}	
\end{corollary}	

Now, let $\Om_0$  be another Polish space, $C_b(\Om_0)$ the space of all bounded continuous functions $\Om_0\to \R$, and $\beta_0$ denote the mixed topology on $C_b(\Om_0)$.\ In order to avoid confusion, we write $C_b(\Om)$ instead of $C_b$ in the following corollary.
\begin{corollary}\label{cor.lcs}
Let $U\colon C_b(\Om)\to C_b(\Om_0)$ be convex and monotone.\ Then, the following statements are equivalent.
\begin{enumerate}
	\item[(i)] $U$ is $\beta$-$\beta_0$-continuous.
	\item[(ii)] For every $\om\in \Omega_0$ and every sequence $(f_n)_{n \in \N}\subset C_b(\Om)$ with $f_n\searrow 0$ as $n\to \infty$,
	\[
	\lim_{n\to \infty}\big(U(f_n)\big)(\om)=\big(U(0)\big)(\om).
	\]
	\item[(iii)] For every compact $K\subset \Omega_0$, every constant $r\geq 0$, and all $\ep >0$, there exists a compact $C\subset \Omega$ and a constant $M\geq 0$ such that
	\[
	\sup_{\om\in K}\big|\big(U(f_1)\big)(\om)-\big(U(f_2)\big)(\om)\big|\leq M\sup_{x\in C}|f_1(x)-f_2(x)|+\ep
	\] 
	for all $f_1,f_2\in C_b(\Om)$ with $\max\big\{\|f_1\|_\infty,\|f_2\|_\infty\big\}\leq r$.
\end{enumerate}	
\end{corollary}	

\section{Proofs}\label{sec.proofs}

\mn{
Before turning our focus on the proof of Theorem \ref{thm.main}, we collect some well-known facts on the connection between pointwise monotone convergence, uniform convergence on compacts together with uniform boundedness, and convergence in the mixed topology for sequences in $C_b$.

\begin{remark}\label{rem.convergenceprop}\
\begin{enumerate}
\item[a)] Let $(f_n)_{n\in \N}\subset C_b$ with $f_n\nearrow f\in C_b$ as $n\to \infty$.\ Since $f_n\nearrow f$ as $n\to \infty$, it follows that $(f_n-f)_{n\in \N}$ is uniformly bounded, i.e., 
\begin{equation}\label{eq.unifbddnes}
\sup_{n\in \N}\|f_n-f\|_\infty<\infty,
\end{equation}
and, by Dini's lemma, $(f_n-f)_{n\in \N}$ converges uniformly on compacts to $0\in C_b$, i.e., for all compacts $C\subset \Om$,
\begin{equation}\label{eq.unifconvcompacts}
\sup_{x\in C}|f_n(x)-f(x)|\to 0\quad\text{as }n\to \infty.
\end{equation}
\item[b)] Let $(f_n)_{n\in \N}\subset C_b$ and $f\in C_b$ with \eqref{eq.unifbddnes} and \eqref{eq.unifconvcompacts}. Moreover, let $V\in \mathcal V$. Then, there exist $(V_k^1)_{k\in \N}\subset \mathcal V^1$ and $V^2\in \mathcal V^2$ such that $$V=\bigcup_{n\in \N}\sum_{k=1}^n (V_k^1\cap k V^2).$$
By \eqref{eq.unifbddnes}, there exists some $k\in \N$ with $f_n-f\in kV^2$ for all $n\in \N$. Moreover, by \eqref{eq.unifconvcompacts}, there exists some $n_0\in \N$ such that $f_n-f\in V_k^1$ for all $n\in \N$ with $n\geq n_0$. Hence, $f_n-f\in V_k^1\cap kV^2\subset V$ for all $n\in \N$ with $n\geq n_0$. This shows that $f_n\to f$ as $n\to \infty$ in the mixed topology $\beta$.
\end{enumerate}
\end{remark}}

\begin{proof}[Proof of Theorem \ref{thm.main}]
	\mn{Let $d\colon \Om\times \Om\to [0,\infty)$ be a metric that is consistent with the topology on $\Om$ such that $(\Om,d)$ is a complete separable metric space.\ Moreover, let $(x_i)_{i\in \N}\subset \Om$ be a sequence such that $D:=\{x_i \, |\, i\in \N\}$ is dense in $\Om$.
	
	We start with the proof of the implication (i)$\, \Rightarrow\,$(ii). To that end, assume that $U$ is continuous from below.\ We show that $U$ is lower semicontinuous in the mixed topology $\beta$, i.e., for all $f\in C_b$ and $\ep>0$, there exists some $V\in \mathcal V$ such that
    \[
    U(f)\leq U(f+e)+\ep\quad\text{for all }e\in V.
    \]
    In order to do so, fix $f\in C_b$ and $\ep>0$. Since $U$ is continuous from below, there exists some $\de>0$ such that
	\[
	U(f)\leq U(f-\de)+\frac\ep2.
	\]
    In a first step, we adapt the main idea from the proof of Ulam's theorem, cf.\ \cite[Proof of Theorem 7.1.4]{dudley} to our setting, and recursively construct families of finite sets $D_k^m\subset D$ and continuous functions $\ph_k^m\colon \Om\to [0,1]$, which are indexed by $k,m\in \N$ with $k\leq m$ and satisfy the following three properties:
    \begin{enumerate}
    \item[(P1)] For all $k,m\in \N$ with $k\leq m$,
    \[
    D_k^m\subset D_k^{m+1}\subset \bigcap_{j=k}^m \bigcup_{x\in D_k^j} B\Big(x,\tfrac1j\Big)=:B_k^m.
    \]
    \item[(P2)] For all $k,m\in \N$ with $k\leq m$ and $x\in \Omega \setminus B_k^m$, it holds $\ph_k^m(x)=0$.
    \item[(P3)] For all $m\in \N$, \[U(f-\delta)\leq U\bigg(f-\de -\frac{m(m+1)}{2}+\sum_{k=1}^m k\ph_k^{m}\bigg)+\frac{\ep}2.\]
    \end{enumerate}
  In order to simplify notation, let $l(0):=1$ and set
  \[
f_m:=f-\de-\frac{m(m+1)}{2},\quad \ph_m^{m-1}:=1,\quad\text{and}\quad B^{m-1}_m:=\Omega\quad\text{for all }m\in \N.
  \]
  Now, let $m\in \N$ and define
  $$
  \ps_{k}^{l,m}(x):=\Big(1-m\dist\big(x,\{x_1,\ldots, x_{l}\}\cap B_k^{m-1}\big)\Big)\vee 0
  $$
  for all $x\in \Omega$, $k\in \{1,\ldots, m\}$, and $l\in \N$.\ Since $B_k^{m-1}$ is open, $\ph_k^{m-1}\geq 0$, and $\ph_k^m(x)=0$ for all $x\in \Om\setminus B_k^{m-1}$ for all $k\in \{1,\ldots, m\}$, it follows that
  \[
   \ph_k^{m-1} \ps_{k}^{l,m}\nearrow \ph_k^{m-1}\quad\text{as }l\to \infty
  \]  
  for all $k\in \{1,\ldots, m\}$.\ Using the continuity from below of $U$, there exists some $l(m)\in \N$ with $l(m)\geq l(m-1)$ and
   \[
  U\bigg(f_m+\sum_{k=1}^m k\ph_k^{m-1}\bigg)\leq U\bigg(f_m+\sum_{k=1}^m k\ph_k^{m-1}\psi_{k}^{l(m),m}\bigg)+\frac\ep2 2^{-m},
  \]
  and we define
  \begin{equation}\label{eq.def.sequence}
  \ph_k^m:=\ph_k^{m-1}\ps_{k}^{l(m),m} \quad\text{and}\quad  D_k^m:=\{x_1,\ldots, x_{l(m)}\}\cap B_k^{m-1}.
  \end{equation}
  
  We verify that the sequence constructed this way satisfies the properties (P1) - (P3). By definition, $D_k^m\subset B_k^{m-1}$ and since
  $$
  B_k^m= B_k^{m-1}\cap\Bigg(\bigcup_{x\in D_k^m}B\big(x,\tfrac1m\big)\Bigg),
  $$
  it follows that $D_k^m\subset B_k^m$ for all $m\in \N$. Since $l(m)\leq l(m+1)$ and $B_k^m\subset B_k^{m-1}$, we find that
  \[
  D_k^m=D_k^m\cap B_k^m=\{x_1,\ldots,x_{l(m)}\}\cap B_k^m\subset \{x_1,\ldots,x_{l(m+1)}\}\cap B_k^m= D_k^{m+1}
  \]
  for all $m\in \N$. By \eqref{eq.def.sequence}, $$\ps_k^{l(m),m}(x)=\big(1-m\dist(x,D_k^m)\big)\vee 0$$
  for all $x\in \Om$, $k\in \{1,\ldots, m\}$ and $m\in \N$. Hence,
  \[
  \ph_k^m(x)=\prod_{j=k}^m \ps_k^{l(j),j}(x)=0\quad\text{for all }x\in \Om\setminus B_k^m,\; k\in \{1,\ldots, m\}, \text{ and }m\in \N.
  \]
  Moreover, setting $f_0:=f-\de$ and using the fact that $\ph_m^{m-1}=1$, it follows that
\[
U\bigg(f_{m-1}+\sum_{k=1}^{m-1} k\ph_k^{m-1}\bigg)=U\bigg(f_m+\sum_{k=1}^m k\ph_k^{m-1}\bigg)\leq U\bigg(f_m+\sum_{k=1}^m k\ph_k^m\bigg)+\frac{\ep}{2}2^{-m}
\]
  for all $m\in \N$. Inductively, we obtain that
  \[
  U(f-\de)\leq U\bigg(f-\de-\frac{m(m+1)}{2}+\sum_{k=1}^m k\ph_k^m\bigg)+\frac{\ep}{2}\quad\text{for all }m\in \N.
  \]
  We have therefore verified the properties (P1) - (P3), and define
\[
C_k:=\bigcap_{j=k}^\infty\bigcup_{x\in D_k^j} \overline B\big(x,\tfrac1j \big)\quad\text{for all }k\in \N,
\]	
where, for $x\in \Om$ and $r>0$, $\overline B(x,r):=\{y\in \Om\, |\, d(x,y)\leq r\}$.  Then, $C_k$ is a closed and totally bounded subset of a complete metric space, hence compact for all $k\in \N$. Observe that
\[
\bigcap_{m=k}^\infty B_k^m= \bigcap_{m=k}^\infty\bigcap_{j=k}^m \bigcup_{x\in D_k^j} B\big(x,\tfrac1j\big)=\bigcap_{j=k}^\infty \bigcup_{x\in D_k^j} B\big(x,\tfrac1j\big)\subset  C_k\quad\text{for all }k\in \N,
\]
so that, by Property (P1), $D_k^m\subset C_k$ for all $k,m\in \N$ with $k\leq m$.

Using the sequence $(C_k)_{k\in \N}$ of compacts, we define
\[
V_k^1:=\Big\{e\in C_b\, \Big|\, \sup_{x\in C_k}|e(x)|< 2^{-k}\de\Big\}\quad \text{for all }k\in \N
\]
and $V^2:=\{g\in C_b(\Om)\, |\, \|g\|_\infty<1\}$. Then, $$V:=\bigcup_{n\in \N}\sum_{k=1}^n\big(V_k^1\cap kV^2\big)$$
 is a neighborhood of $0\in C_b$ in the mixed topology.
 
 We show that
 \[
  U(f)\leq U(f+e)+\ep\quad\text{for all }e\in V.
 \]
 To that end, let $e\in V$, i.e., there exist $n\in \N$ and $e_k\in V_k^1\cap kV^2$ for $k\in \{1,\ldots,n\}$ such that $e=\sum_{k=1}^n e_k$. Let $\delta_k:=\sup_{x\in C_k}|e_k(x)|$. Since $e_k\in V_k^1$, it follows that $\delta_k<2^{-k}\delta$ for all $k\in \{1,\ldots,n\}$.
 As $C_1,\ldots, C_n$ are compact,
 there exists some $m\in \N$ with $m\geq n$ such that
 \[
  	|e_k(y)-e_k(x)|< 2^{-k}\de-\delta_k
 \]
 for all $k\in \{1,\ldots, n\}$, $x\in C_k$, and $y\in \Om$ with $d(x,y)<\frac1m$.\ Since $D_k^m\subset C_k$ for all $k\in \{1,\ldots,n\}$, this implies that
 \[
 |e_k(y)|\leq 2^{-k}\de\quad \text{for all }\mn{k\in \{1,\ldots, n\}\text{ and }} y\in \bigcup_{x\in D_k^m} B\big(x,\tfrac1m \big).
 \]
 Hence, by Property (P2), it follows that
 \[
 -k+k\ph_k^m\leq -k+(k+e_k)\ph_k^m+2^{-k}\de \leq  e_k+ 2^{-k}\delta\quad\text{for all }k\in \{1,\ldots,n\},
 \]
 where, in the first step, we used the fact that $|e_k\ph_k^m|\leq 2^{-k}\de$ and, in the last step, we used the fact that $k+e_k\geq 0$ for all $k\in \{1,\ldots,n\}$.\ On the other hand, $-k+k\ph_k^m\leq 0$ for all $k\in \{n+1,\ldots, m\}$. Therefore,
  \[
 f-\de-\frac{m(m+1)}{2}+\sum_{k=1}^m k\ph_k^m\leq f-\de+\sum_{k=1}^n\big(e_k+2^{-k}\de\big)\leq f+\sum_{k=1}^ne_k=f+e.
 \]
 Using Property (P3), it follows that
 \[
 U(f)\leq U(f-\de)+\frac\ep2\leq U\bigg(f-\de-\frac{m(m+1)}{2}+\sum_{k=1}^m k\ph_k^m\bigg)+\ep\leq U(f+e)+\ep.
 \]
 This proves that $U$ is lower semicontinuous w.r.t.\ the mixed topology $\beta$.

Next, we prove that (ii) implies (iv). To that end, let $f\in C_b$ and $\ep>0$. Then, there exist $(V_k^1)_{k\in \N}\subset \mathcal V^1$ and $V^2\in \mathcal V^2$ such that
\[
U(f)\leq U(f+e)+\ep
\]
for all $e\in V:=\bigcup_{n\in \N}\sum_{k=1}^n(V_k^1\cap kV^2)$. Let $r\geq0$. Then, there exists some $n\in \N$ such that
\[
 \{g\in C_b\, |\, \|g\|_\infty\leq r\}\subset nV^2.
\]
Moreover, there exist $\de>0$ and a compact $C\subset \Omega$ such that
\[
\Big\{g\in C_b\, \big|\, \sup_{x\in C}|g(x)|\leq\delta\Big\}\subset V_n^1.
\]
Hence, for every $e\in C_b$ with $\|e\|_\infty\leq r$ and $\sup_{x\in C}|e(x)|\leq \delta$, it follows that $e\in V$ and therefore $U(f)\leq U(f+e)+\ep$.

Clearly, (ii) implies (iii), so that it remains to prove the implications (iii)$\,\Rightarrow\,$(i) and (iv)$\,\Rightarrow\,$(i).\ To that end, let $(f_n)_{n\in \N}\subset C_b$ with $f_n\nearrow f\in C_b$ as $n\to \infty$. Then, due to the monotonicity of $U$
\[
\lim_{n\to \infty} U(f_n)=\sup_{n\in \N}U(f_n)\leq U(f).
\]
Since $f_n\nearrow f$ as $n\to \infty$, by Remark \ref{rem.convergenceprop} a), it follows that the sequence $(f_n-f)_{n\in \N}$ is uniformly bounded and converges uniformly on compacts to $0\in C_b$. Moreover, by Remark \ref{rem.convergenceprop} b), it follows that $f_n\to f$ as $n\to \infty$ in the mixed topology $\beta$.

Hence, if $U$ satisfies (iii) or (iv), it follows that
\[
U(f)\leq \liminf_{n\to \infty}U(f_n)=\lim_{n\to \infty}U(f_n).
\]
The proof is complete.}
\end{proof}

\begin{proof}[Proof of Corollary \ref{cor.convex}]
	The equivalence of (i) and (ii) follows from Theorem \ref{thm.main}.\ \mn{By standard duality theory in locally convex Hausdorff spaces, cf.\ \cite{ekeland}, (ii) is equivalent to the fact that $U$ admits a dual representation of the form 
    \[
    U(f)=\sup_{\mu\in \mathcal M}\big(\mu f-\alpha(\mu)\big)\quad\text{for all }f\in C_b
    \]
 with a set $\mathcal M$ of $\beta$-continuous linear functionals on $C_b$ and a function $\alpha\colon \mathcal M\to \R$. Let $\mu\in \mathcal M$ and $f\in C_b$ with $f\geq 0$. Since $U$ is monotone, it follows that
 \[
 \frac1\lambda\big(\mu (-\lambda f)-\alpha(\mu)\big)\leq \frac{U(-\lambda f)}\lambda\leq \frac{U(0)}{\lambda}\quad\text{for all }\lambda>0.
 \]
 Hence,
\[
 0=\lim_{\lambda\to \infty} -\frac{U(0)+\alpha (\mu)}\lambda\leq \mu f,
\] 
which shows that every linear functional in $\mathcal M$ is positive.}
 The remaining equivalences, in particular, the dual representation \eqref{eq.dual} now follow from the fact that, by Theorem \ref{thm.main} and the Daniell-Stone theorem, cf.\ \cite[Theorem 7.8.1]{bogachev}, a positive linear functional $\mu\colon C_b\to \R$ is continuous in the mixed topology $\beta$ if and only if it belongs to $\ca_+$.
\end{proof}

\begin{proof}[Proof of Corollary \ref{cor.choquet}]
	 \mn{We first prove the implication (i)$\,\Rightarrow\,$(ii).  Let $(f_n)_{n\in \N}\subset L_b$ and $f\in L_b$ with $f_n\nearrow f$ as $n\to \infty$. Then, for all $s\in \R$,
	\[
	\lim_{n\to \infty}c(f_n>s)= c(f>s).
	\]
	Using the monotone convergence theorem, it follows that
	\begin{align}
		\notag\lim_{n\to \infty} \int_\Om f_n\, {\rm d}c&=\lim_{n\to \infty} \bigg(\int_0^\infty c(f_n>s)\, {\rm d}s+\int_{-\infty}^0 \big(c(f_n>s)-c(\Om)\big)\,{\rm d}s\bigg)\\
  & =\int_0^\infty c(f>s)\, {\rm d}s+\int_{-\infty}^0 \big(c(f>s)-c(\Om)\big)\,{\rm d}s=\int_\Om f\, {\rm d}c. \label{eq.proof.choquet}
	\end{align}
    For the implication (ii)$\,\Rightarrow\,$(iii), first observe that
    \[
    c(B)\geq \sup_{C\Subset B}\inf_{\substack{A \in \mathcal O\\ A\supset C}}c(A)\quad\text{for all }B\in \mathcal O.
    \]
    In order to show the converse inequality, let $B\in \mathcal O$ and $\ep>0$. In a first step, we consider the case $B=\Om$. Then, by Theorem \ref{thm.main}, there exist $\de>0$ and a compact set $C\subset \Omega$ such that
    \[
    \int_\Om 1\, {\rm d}c\leq \int_\Om g\, {\rm d}c+\ep
    \]
    for all $g\in C_b$ with $\|g\|_\infty\leq 1$ and $\sup_{x\in C} |g(x)-1|\leq \de$. Now, let $A\in \mathcal O$ with $A\supset C$. Since $A$ is open, there exists some $m\in \N$ such that $g\colon \Om\to \R$, given by
    \[
    g(x):=\sup_{y\in C} \big(1- md(x,y)\big)\vee 0\quad\text{for all }x\in \Om,
    \]
    satisfies $g(x)=0$ for $x\in \Om\setminus A$. Since $0\leq g\leq 1$ and $g(x)=1$ for all $x\in C$, it follows that
    \[
    c(\Om)=\int_\Om 1\, {\rm d}c\leq  \int_\Om g\, {\rm d}c+\ep\leq \int_\Om 1_A\, {\rm d}c+\ep= c(A)+\ep.
    \]
    We have therefore shown that
    \[
    c(\Om)\leq \inf_{\substack{A \in \mathcal O\\ A\supset C}}c(A)+\ep
    \]
    Taking the supremum over all $C\Subset \Om$ and letting $\ep\to 0$, Equation \eqref{eq.regularity} follows.

     For general $B\in \mathcal O$, the statement now follows from the fact that $B$, endowed with the subspace topology
    \[
    \mathcal O_B:=\{A\cap B\,|\, A\in \mathcal O\}=\{A\in \mathcal O\,|\, A\subset B\}\subset \mathcal O,
    \]
    is again a Polish space together with the observation that a subset of $B$ is compact in the subspace topology $\mathcal O_B$ if and only if it is compact in the original topology $\mathcal O$.

Next, we prove that (iii) implies (i). To that end, let $(B_n)_{n\in \N}\subset \mathcal O$ with $B_n\subset B_{n+1}$ for all $n\in \N$ and $\ep>0$. Then, there exists some compact $C\subset \bigcup_{n\in \N} B_n=:B$ with
    \[
    c(B)\leq \inf_{\substack{A \in \mathcal O\\ A\supset C}}c(A)+\ep.
    \]
    Since $C$ is compact, $C\subset B=\bigcup_{n\in \N} B_n$, and $B_n$ is open with $B_n\subset B_{n+1}$ for all $n\in \N$, there exists some $n_0\in \N$ such that $C\subset B_{n_0}$. Hence, $$\inf_{\substack{A\in \mathcal O\\ A\supset C}}c(A)\leq c(B_{n_0}),$$ and it follows that
    \[
    \sup_{n\in \N} c(B_n)\leq c(B)\leq \sup_{n\in \N} c(B_n)+\ep=\lim_{n\to \infty} c(B_n)+\ep.
    \]
    Letting $\ep\to 0$, it follows that $c(B)=\lim_{n\to \infty} c(B_n)$.

Now, we assume that the capacity $c$ is $2$-alternating. In order to prove that (iv)$\,\Rightarrow\,$(i), let $(B_n)_{n\in \N}\subset \mathcal O$ with $B_n\subset B_{n+1}$ for all $n\in \N$. Then,
\begin{align*}
c\bigg(\bigcup_{n\in \N} B_n\bigg)=\sup_{\mu\in \mathcal M} \mu\bigg(\bigcup_{n\in \N} B_n\bigg)=\sup_{\mu\in \mathcal M}\sup_{n\in \N} \mu (B_n)=\sup_{n\in \N}\sup_{\mu\in \mathcal M} \mu (B_n)= \lim_{n\to \infty}c(B_n).
\end{align*}
    
	In a last step, we prove that (ii) implies (iv).\ By Corollary \ref{cor.convex}, there exists a set $\mathcal M\subset \ca_+$ and a function $\alpha\colon \mathcal M\to \R$ such that
	\[
	\int_\Omega f\, {\rm d}c=\sup_{\mu\in \mathcal M}\bigg(\int_\Omega f\, {\rm d}\mu-\alpha(\mu)\bigg)\quad\text{for all }f\in C_b.
	\]
    Since the Choquet integral is positively homogeneous, it follows that
    \[
    \int_\Om f\, {\rm d}\mu-\frac{\alpha(\mu)}{\lambda}=\frac1\lambda\bigg(\int_\Om \lambda f\, {\rm d}\mu-\alpha (\mu)\bigg)\leq \frac1\lambda \int_\Om \lambda f\, {\rm d}c=\int_\Om f\, {\rm d}c
    \]
    for all $f\in C_b$, $\mu\in \mathcal M$, and $\lambda>0$. Letting $\lambda\to \infty$, it follows that
    $$
    \sup_{\mu\in \mathcal M}\int_\Om f\, {\rm d}\mu\leq \int_\Om f\, {\rm d}c\quad\text{for all }f\in C_b.
    $$
    On the other hand,
    \[
    -\alpha(\mu)=\int_\Om 0\, {\rm d}\mu-\alpha(\mu)\leq \int_\Omega 0\, {\rm d}c=0\quad \text{for all }\mu\in \mathcal M.
    \]
    Hence,
    \[
    \sup_{\mu\in \mathcal M}\int_\Om f\, {\rm d}\mu\leq \int_\Omega f\, {\rm d}c=\sup_{\mu\in \mathcal M}\bigg(\int_\Omega f\, {\rm d}\mu-\alpha(\mu)\bigg)\leq \sup_{\mu\in \mathcal M}\int_\Omega f\, {\rm d}\mu \quad\text{for all }f\in C_b.
    \]
    In particular,
    $$
    \mu(\Omega)=\int_\Om 1\,{\rm d}\mu\leq \int_\Om 1\,{\rm d}c=c(\Om)\quad\text{for all }\mu\in \mathcal M.
    $$
	Since the Choquet integral is constant additive, it follows that $$0=c(\Om)+\int_\Om (-1)\, {\rm d}c\geq c(\Om)+\int_\Om (-1)\, {\rm d}\mu=c(\Om)-\mu(\Om)$$
 for all $\mu\in \mathcal M$. We have therefore shown that $\mu(\Om)=c(\Om)$ for all $\mu\in \mathcal M$.
 Now, for every $f\in L_b$, there exists a sequence $(f_n)_{n\in \N}\subset C_b$ with $f_n\nearrow f$ as $n\to \infty$,\footnote{Define, for $x\in \Om$ and $n\in\N$, $f_n(x):=\inf_{y\in \Om}\big(f(y)+nd(x,y)\big)$ with a metric $d$ that is consistent with the topology on $\Om$.} so that
	\[
	\int_\Om f\, {\rm d}c=\sup_{n\in \N} \int_\Om f_n\, {\rm d}c= \sup_{n\in \N}\sup_{\mu\in \mathcal M}\int_\Omega f_n\, {\rm d}\mu=\sup_{\mu\in \mathcal M}\sup_{n\in \N}\int_\Omega f_n\, {\rm d}\mu =\sup_{\mu\in \mathcal M}\int_\Omega f\, {\rm d}\mu,
	\]
	where, in the first equality, we used \eqref{eq.proof.choquet} and, in the last equality, we used the monotone convergence theorem.}
\end{proof}	

\begin{proof}[Proof of Theorem \ref{thm.supremum}]
	\mn{First, observe that, by convexity of $U_i$ for all $i\in I$ and \eqref{eq.ass.thm.supremum}, 
 \begin{equation}\label{eq.apriori_ui}
    \sup_{i\in I} \big(U_i(f_1)-U_i(f_2)\big)\leq \sup_{i\in I}\big(U_i(f_1)+U_i(-f_2)-2U_i(0)\big)\leq 2\sup_{i\in I} \big(U_i(r)-U_i(0)\big)<\infty
 \end{equation}
for all $r\geq 0$ and $f_1,f_2\in C_b$ with $\max\big\{\|f_1\|_\infty,\|f_2\|_\infty\big\}\leq r$.}\

The implication (iii)$\,\Rightarrow\,$(i) follows from Remark \ref{rem.convergenceprop} a).\ In order to prove that (i) implies (ii), we first show that, for every $f\in C_b$ and every $\ep >0$, there exists some $V\in \mathcal V$ with 
	\[
	\sup_{i\in I} |U_i(f+e)-U_i(f)|<\ep \quad \text{for all }e\in V.
	\]
	To that end, let $f\in C_b$ and consider the monotone maps $\overline U_f,\underline U_f\colon C_b\to \R$, given by
	\begin{align*}
		\overline U_f(g)&:= \sup_{i\in I} \big(U_i(f+g)- U_i(f)\big)\quad\text{and}\\	
		\underline U_f(g)&:= \sup_{i\in I} \big(U_i(f)-U_i(f-g)\big)\quad\text{for all }g\in C_b.
	\end{align*}
	\mn{Observe that, by \eqref{eq.apriori_ui}, $\overline U_f$ and $\underline U_f$ are well-defined and
 \[
 \underline U_f(-g)= \sup_{i\in I} \big(U_i(f)-U_i(f+g)\big)\quad\text{for all }g\in C_b.
 \]
 Moreover, for any $V\in \mathcal V$, $e\in V$ if and only if $-e\in V$. Hence, }the auxiliary statement follows from Corollary \ref{cor.main}, once we have shown that, both, $\overline U_f$ and $\underline U_f$ are continuous from above.\ To that end, let $g\in C_b$ and $(g_n)_{n\in \N}$ with $g_n\searrow g$ as $n\to \infty$. Moreover, let $\ep>0$, $n\in \N$, and $\la\in (0,1)$. Then, using the fact that $U_i$ is convex for all $i\in I$, 
	\begin{align*}
		\overline U_f(g_n)-\overline U_f(g)&\leq\sup_{i\in I} \big(U_i(f+g_n)-U_i(f+g)\big)\leq \la\sup_{i\in I} \bigg[U_i\bigg(\frac{g_n-g}{\la}\bigg)-U_i(0)\bigg]\\
		&+\la \sup_{i\in I}\big(U_i(0)-U_i(f+g)\big)+(1-\lambda)\sup_{i\in I} \bigg[U_i\bigg(\frac{f+g}{1-\lambda}\bigg)-U_i(f+g)\bigg]
	\end{align*}	
	and
	\begin{align*}
		\underline U_f(g_n)-\underline U_f(g)&\leq \sup_{i\in I} \big(U_i(f-g)-U_i(f-g_n)\big)\leq \sup_{i\in I} \big(U_i(f-2g+g_n)-U_i(f-g)\big)\\
		&\leq  \la\sup_{i\in I} \bigg[U_i\bigg(\frac{g_n-g}{\la}\bigg)-U_i(0)\bigg]+\la \sup_{i\in I}\big(U_i(0)-U_i(f-g)\big)\\
		&+(1-\lambda)\sup_{i\in I} \bigg[U_i\bigg(\frac{f-g}{1-\lambda}\bigg)-U_i(f-g)\bigg]. 	
	\end{align*}
	Since the maps $$\mn{\R\to \R,\; \gamma\mapsto \sup_{i\in I}\Big(U_i\big(\gamma (f\pm g)\big)-U_i(f\pm g)\Big)}$$ are convex and therefore continuous,
	\[
	\la \sup_{i\in I}\big(U_i(0)-U_i(f\pm g)\big)+(1-\lambda)\sup_{i\in I} \bigg[U_i\bigg(\frac{f\pm g}{1-\lambda}\bigg)-U_i(f\pm g)\bigg]<\frac\ep2
	\]
	for $\lambda\in (0,1)$ sufficiently small. Moreover, by assumption,
	\[
	\la\sup_{i\in I} \bigg[U_i\bigg(\frac{g_n-g}{\la}\bigg)-U_i(0)\bigg]<\frac\ep2
	\]
	for $n\in \N$ sufficiently large, since $\frac{g_n-g}{\lambda}\searrow 0$ as $n\to \infty$. We have therefore shown that
	\[
	0\leq \overline U_f(g_n)-\overline U_f(g)<\ep \quad \text{and}\quad 0\leq \underline U_f(g_n)-\underline U_f(g)<\ep
	\]
	for $n\in \N$ sufficiently large, so that both $\overline U_f$ and  $\underline U_f$ are continuous from above. We have therefore proved the auxiliary statement and are now ready to prove the implication (i)$\,\Rightarrow\,$(ii). For $i\in I$ and $f\in C_b$, let $U_{i,f}\colon C_b\to \R$ be given by
	\[
	U_{i,f}(g):= U_i(f+g)-U_i(f)\quad\text{for all }g\in C_b. 
	\]
	Then, $U_{i,f}$ is convex and monotone with $U_{i,f}(0)=0$ for all $i\in I$ and $f\in C_b$. Moreover, for all $\lambda\in (0,1)$, $i\in I$, and $f,g\in C_b$,
	\begin{align*}
	U_{i,f}(g)&\leq \lambda \Bigg(U_i\bigg(\frac{g}{\lambda}\bigg)-U_i(f)\Bigg)+(1-\lambda)\Bigg(U_i\bigg(\frac{f}{1-\lambda}\bigg)-U_i(f)\Bigg)\\
& \leq \lambda \Bigg(U_i\bigg(\frac{g}{\lambda}\bigg)-U_i(0)\Bigg)+\lambda \big(U_i(-f)-U_i(0)\big)\\
&\quad +(1-\lambda)\Bigg(U_i\bigg(\frac{f}{1-\lambda}\bigg)-U_i(f)\Bigg),
	\end{align*}
\mn{where, in the second inequality, we used the fact that $U_i(0)-U_i(f)\leq U_i(-f)-U_i(0)$ for all $i\in I$.\ Let $r\geq 0$, $(g_n)_{n\in \N}\subset C_b$ with $g_n\searrow 0$ as $n\to \infty$, and $\ep>0$.\ Then, by \eqref{eq.apriori_ui}, the map
\[
\R\to \R,\; \gamma\mapsto \sup_{\|f\|_\infty\leq r}\sup_{i\in I} \big(U_i(\gamma f)-U_i(f)\big)
\]
 is convex and well-defined. Therefore, it is continuous and it follows that
 \[
 \lambda \big(U_i(-f)-U_i(0)\big) +(1-\lambda)\Bigg(U_i\bigg(\frac{f}{1-\lambda}\bigg)-U_i(f)\Bigg)< \frac\ep2
 \]
 for $\lambda\in (0,1)$ sufficiently small. Hence,
\[
\sup_{\|f\|_\infty\leq r}\sup_{i\in I}U_{i,f}(g_n)\leq \lambda \sup_{i\in I}\Bigg(U_i\bigg(\frac{g_n}{\lambda}\bigg)-U_i(0)\Bigg)+\frac\ep2<\ep
\]
for $n\in \N$ sufficiently large, and we have shown that
	\[
	\lim_{n\to \infty}\sup_{\|f\|_\infty\leq r}\sup_{i\in I}U_{i,f}(g_n)=0.
	\]
	Using the auxiliary statement for the convex monotone functions $U_{i,f}$ with $i\in I$ and $f\in C_b$ with $\|f\|_\infty\leq r$}, there exists some $V\in \mathcal V$ such that
	\[
	\sup_{\|f\|_\infty\leq r}\sup_{i\in I}|U_i(f+e)-U_i(f)|=\sup_{\|f\|_\infty\leq r}\sup_{i\in I}|U_{i,f}(e)|\leq \ep\quad\text{for all }e\in V.
	\]
	It remains to prove the implication (ii)$\,\Rightarrow\,$(iii). Let $r\geq 0$ and $\ep>0$.  Then, there exists some $V\in \mathcal V$ such that
	\[
	\sup_{\|f\|_\infty\leq r}\sup_{i\in I}|U_i(f+e)-U_i(f)|\leq \ep\quad\text{for all }e\in V.
	\]
	By definition of the local base $\mathcal V$, there exists some compact set $C\subset \Omega$ and some $\delta>0$ such that
	\[
	 \Big\{e\in C_b\,\Big|\, \sup_{x\in C} |e(x)|< \delta\Big\} \cap \{e\in C_b\, |\, \|e\|_\infty\leq 2r\}\subset V.
	\]
	Now, let $f_1,f_2\in C_b$ with $\max\big\{\|f_1\|_\infty,\|f_2\|_\infty\big\}\leq r$. \mn{Then, by the triangle inequality, $\|f_1-f_2\|_\infty\leq 2r$.} If $\sup_{x\in C} |f_1(x)-f_2(x)|< \delta$, then
	\[
	\sup_{i\in I}|U_i(f_1)-U_i(f_2)|\leq \ep.
	\]
	On the other hand, if $\sup_{x\in C} |f_1(x)-f_2(x)|\geq \delta$, \mn{then, by \eqref{eq.apriori_ui},} it follows that
	\[
	\sup_{i\in I}|U_i(f_1)-U_i(f_2)|\leq M\sup_{x\in C} |f_1(x)-f_2(x)|
	\]
	with $M:=\frac{2}{\de}\sup_{i\in I}\big(U_i(r)-\mn{U_i(0)}\big)$. The proof is complete.
\end{proof}	

\begin{proof}[Proof of Corollary \ref{thm.lcs}]
	Since $(L,\tau)$ is a locally convex vector lattice, (iv) implies (i). The implication (i)$\,\Rightarrow\,$(iii) is trivial and, by Corollary \ref{cor.main}, (ii) and (iii) are equivalent. By Theorem \ref{thm.supremum} and Lemma \ref{lem.lcs}, (ii)$\,\Rightarrow\,$(iv) and (ii)$\,\Rightarrow\,$(v), since (ii) together with Dini's lemma implies that, for every convex and weak$^*$ compact set $K$ of nonnegative $\tau$-continuous linear functionals, the map
	\[
	U_K\colon C_b\to \R, \;f\mapsto \sup_{\mu\in K} \mu\big(U(f)-U(0)\big)
	\]
	is convex and satisfies $\lim_{n\to \infty}U_K(f_n)= 0$ for every sequence $(f_n)_{n\in \N}\subset C_b$ with $f_n\searrow 0$ as $n\to \infty$. Since, for every nonnegative $\tau$-continuous linear functional $\lambda\colon L\to \R$, there exists a $\tau$-continuous lattice seminorm $p\colon L\to [0,\infty)$ with $|\la u|\leq p(u)$ for all $u\in L$, (v) implies (ii) by Dini's lemma.
\end{proof}

\appendix

\section{Capacities and Choquet integrals}\label{app.capacity}
The setup and notation in this section follows the one of the main part. \mn{The following lemma is sort of a folklore theorem, cf.\ \cite[Properties 11.8 and Theorem 11.11]{MR1633615}.\ For the reader's convenience, we nevertheless provide a short proof.} 

\begin{lemma}
Let $c\colon \mathcal O\to [0,\infty)$ be a capacity.\ Then, the following statements are equivalent.
\begin{enumerate}
	\item[(i)] For all $B_1,B_2\in \mathcal O$,
	\[
	c(B_1\cup B_2)+c(B_1\cap B_2)\leq c(B_1)+c(B_2).
	\]
	\item[(ii)] For all $f_1,f_2\in L_b$,
	\[
	\int_\Om (f_1+f_2)\, {\rm d}c\leq \int_\Om f_1\, {\rm d}c+\int_\Om f_2\, {\rm d}c.
	\]
\end{enumerate}		
\end{lemma}	

\begin{proof}
	We first prove the implication (ii)$\, \Rightarrow\,$(i). To that end, let $B_1,B_2\in \mathcal O$. Then,
	\begin{align*}
	 c(B_1\cup B_2)+c(B_1\cap B_2)&=\int_\Om (1_{B_1}+1_{B_2})\, {\rm d}c\leq  \int_\Om 1_{B_1}\, {\rm d}c+\int_\Om 1_{B_2}\, {\rm d}c\\
	 &=c(B_1)+c(B_2).
	\end{align*}
	We proceed with the proof of the implication (i)$\, \Rightarrow\,$(ii). In a first step, we prove by an induction over $n\in \N$ that
	\begin{equation}\label{eq.induction}
	\int_\Om\sum_{i=1}^n 1_{B_i}\,{\rm d}c\leq \sum_{i=1}^nc(B_i)\quad\text{for all }B_1,\ldots B_n\in \mathcal O.
	\end{equation}
	For $n=1$, the statement is trivial. Therefore, assume that \eqref{eq.induction} is proved for some $n\in \N$, and let $B_1,\ldots B_{n+1}\in \mathcal O$. Then,
	\[
	\sum_{k=1}^{n+1} 1_{B_i}=1_{\bigcup_{i=1}^{n+1}B_i}+\sum_{k=1}^n 1_{\left(\bigcup_{i=1}^kB_i\right)\cap B_{k+1}}.
	\]
	Using \eqref{eq.induction} and (i), we obtain that
	\begin{align*}
    	\int_\Om\sum_{i=1}^{n+1} 1_{B_i}\,{\rm d}c&=c\bigg(\bigcup_{i=1}^{n+1}B_i\bigg)+\int_\Om	\sum_{k=1}^n 1_{\left(\bigcup_{i=1}^kB_i\right)\cap B_{k+1}}\, {\rm d}c\\
    	&\leq c\bigg(\bigcup_{i=1}^{n+1}B_i\bigg)+ \sum_{k=1}^n c\Bigg(\bigg(\bigcup_{i=1}^kB_i\bigg)\cap B_{k+1}\Bigg)\\
    	&\leq \sum_{i=1}^{n+1}c(B_i).
	\end{align*}
	Now, let $f_1,f_2\in L_b$. Since the Choquet integral is constant additive, we may w.l.o.g.\ assume that $f_1\geq 0$ and $f_2\geq 0$. Let $r:=\max\big\{\|f_1\|_\infty,\|f_2\|_{\infty}\}\leq r$. For $i=1,2$, $n\in \N$, and $k=1,\ldots, 2^n$, define $B^k_{i,n}:=\{f_k>k2^{-n}r\}$. Then, for $i=1,2$,
	\[
	\bigg\|f_i-\sum_{k=1}^{2^n}2^{-n}r1_{B^k_{i,n}}\bigg\|_\infty\leq 2^{-n}r\to 0\quad \text{as }n\to \infty.
	\]
	Since, by positive homogeneity of the Choquet integral and \eqref{eq.induction},
	\begin{align*}
	\int_\Om \sum_{k=1}^{2^n}2^{-n}r\big(1_{B^k_{1,n}}+1_{B^k_{2,n}}\big)\, {\rm d}c&= 2^{-n}r\int_\Om \sum_{k=1}^{2^n}\big(1_{B^k_{1,n}}+1_{B^k_{2,n}}\big)\, {\rm d}c\\
	&\leq  2^{-n}r\sum_{k=1}^{2^n} \Big(c\big(B^k_{1,n}\big)+c\big(B^k_{2,n}\big)\Big)\\
	&=	\int_\Om \sum_{k=1}^{2^n}2^{-n}r1_{B^k_{1,n}}\, {\rm d}c+	\int_\Om \sum_{k=1}^{2^n}2^{-n}r1_{B^k_{2,n}}\, {\rm d}c,
	\end{align*}
 it follows that
 \[
  \int_\Om (f_1+f_2)\, {\rm d}c\leq \int_\Om f_1\, {\rm d}c+\int_\Om f_2\, {\rm d}c.
 \]
\end{proof}	

\section{Locally convex vector lattices}\label{app.lcs}

Thoughout this section, let $(L,\tau)$ be a locally convex vector lattice, i.e., a vector lattice $L$ together with a locally convex topology $\tau$ on $L$, which is generated by a family of lattice seminorms. Let $L_+:=\{u\in L\, |\, u\geq 0\}$ and $L'$ be the topological dual space of $L$, i.e., the space of all continuous linear functionals $L\to \R$. Moreover, let
$$L_+':=\big\{\lambda\in L'\, \big|\, \forall u\in L_+\colon \lambda u\geq 0\big\}$$
be the set of all positive continuous linear functionals on $L$. For $u\in L$, we use the standard notation $u_+:=u\vee 0$ and $u_-:=-(u\wedge 0)$. Then, $u=u_+-u_-$ and $|u|:=u_++u_-$ for all $u\in L$. \mn{The following lemma can be deduced from \cite[Theorem 8.24 and Corollary 8.25]{MR1717083} together with the fact that every linear functional that is bounded by a lattice seminorm is order bounded. For the sake of a self-contained exposition, we provide a short proof.}

\begin{lemma}\label{lem.lcs}
	Let $p\colon L\to [0,\infty)$ be a continuous lattice seminorm.
	\begin{enumerate}
    \item[a)]For every $\lambda\in L'$ with $|\la u|\leq p(u)$ for all $u\in L$, there exists $\lambda_+,\lambda_-\in L_+'$ with
      \begin{equation}\label{eq.functional1}
	   \lambda u=\lambda_+u-\lambda_-u\quad\text{and}\quad\max\big\{\la_+|u|,\la_-|u|\}\leq p(u)\quad\text{for all }u\in L.
      \end{equation}
	\item[b)] There exists a convex and weak$^*$ compact set $K\subset L'_+$ with
	\[
	\max_{\mu\in K}|\mu u|\leq \max_{\mu\in K} \mu |u|=p(u)\leq 2\max_{\mu\in K} |\mu u|\quad\text{for all }u\in L.
	\]
	\end{enumerate}	
\end{lemma}	

\begin{proof}
	We start with the proof of part a). Since $p$ is a lattice seminorm, it follows that $p(u)=p(v)$ for all $u,v\in L$ with $|u|=|v|$.\ In particular, 
		\begin{equation}\label{eq.seminorm1}
		p(u)=p(|u|)\quad\text{for all }u\in L.
	\end{equation}
	\mn{Let $\lambda\in L'$ with $|\la u|\leq p(u)$ for all $u\in L$, and define 
	\[
	\lambda_+ u:=\sup\big\{\lambda v\, \big|\, v\in L_+,\, v\leq u\big\}\quad\text{for all }u\in L_+.
	\]
	}Since $p$ is a lattice seminorm, $0\leq \lambda_+ u\leq p(u)$ and $\lambda_+(\alpha u)=\alpha \lambda_+u$ for all $u\in L_+$ and $\alpha\geq 0$. Let $u_1,u_2\in L_+$. Then, for $v_1,v_2\in L_+$ with $v_1\leq u_1$ and $v_2\leq u_2$,
	\[
	\lambda v_1+\lambda v_2\leq \lambda (v_1+v_2)\leq \lambda_+(u_1+u_2).
	\] 
	Hence, $\lambda_+ u_1+\lambda_+u_2\leq \lambda_+(u_1+u_2)$. On the other hand, for $v\in L_+$ with $v\leq u_1+u_2$, \mn{let
    $$v_1:=(v-u_2)^+\geq 0 \quad \text{and} \quad v_2:=v-v_1=v+(u_2-v)\wedge 0=v\wedge u_2\leq u_2.$$
    Moreover, $v_1\leq u_1$ since $u_1\geq 0$ and $v_2= v\wedge u_2\geq 0$ since $v\geq 0$ and $u_2\geq0$.} Hence,
	\[
	\lambda v= \la v_1+\la \mn{v}_2\leq \lambda_+u_1+\lambda_+u_2.
	\]
	We have therefore shown that $\lambda_+(u_1+u_2)=\lambda_+ u_1+\lambda_+u_2$. For $u\in L$, define
	\[
	\lambda_+ u:=\lambda_+ u_+-\lambda_+ u_-.
	\]
	\mn{Let $u,v\in L$. Then, 
    \[
    (u+v)_+-(u+v)_-=u+v=u_+-u_-+v_+-v_-,
    \]
    so that
    \begin{align*}
    \lambda_+(u+v)_++\lambda_+u_-+\lambda_+v_-&=\lambda_+ \big((u+v)_++u_-+v_-\big)=\lambda_+ \big((u+v)_-+u_++v_+\big)\\
    &=\lambda_+(u+v)_-+\lambda_+u_++\lambda_+v_+,
    \end{align*}
    which implies that $\lambda_+(u+v)=\lambda_+u+\lambda_+v$. Moreover, for all $u\in L$ and $\al>0$,
    \[
    \lambda_+(\alpha u)=\lambda_+(\alpha u)_+-\lambda_+(\alpha u)_-=\lambda_+( \alpha u_+)-\lambda_+( \alpha u_-)=\alpha \lambda_+ u.
    \]
    }Since, by definition, $(-u)_+=-u_-$ and $(-u)_-=-u_+$ for all $u\in L$ it follows that $\lambda_+\colon L\to \R$ is linear and, by \eqref{eq.seminorm1},
	\[
	|\lambda_+ u|\leq \lambda_+ |u|\leq p(|u|)=p(u),
	\]
	which implies that $\la_+\colon L\to \R$ is continuous. Now, defining $\la_-u:=\la_+ u-\la u$ for all $u\in L$, we find that
	\begin{align*}
	\lambda_- u=\sup\big\{\lambda (v-u)\, |\, v\in L_+,\, v\leq u\big\}=\sup\big\{\lambda (-w)\, \big|\, w\in L_+,\, w\leq u\big\}=(-\lambda)_+ u
	\end{align*}
	\mn{for all $u\in L_+$. Using the fact that $\lambda_-$ is linear and replacing $\lambda$ by $-\lambda$, it follows that
    \[
    |\lambda_-u|\leq \lambda_-|u|=(-\lambda)_+|u|\leq p(u).    
    \]
    In particular, $\lambda_-$ is continuous, and the proof of part a) is complete.}
	
	In order to prove part b), let $V:=\{u\in L\, |\, p(u)\leq 1\}$.\ By the Banach-Alaoglu theorem, the set $V^\circ:=\{\lambda\in L'\, |\, \forall u\in V\colon \lambda u\leq 1\}$ is  convex and weak$^*$ compact.\ By \eqref{eq.seminorm1} and the Hahn-Banach theorem,
	\[
	p(u)=\max_{\lambda\in V^\circ} |\lambda u|\quad\text{for all }u\in L.
	\]
	Since the set $L_+':=\bigcap_{u\in L_+}\{\la\in L'\, |\, \la u\geq 0\}$ is convex and weak$^*$ closed, it follows that $K:=V^\circ\cap L_+$ is convex and weak$^*$ compact. By \eqref{eq.seminorm1},
    \[
     \max_{\mu\in K}|\mu u|\leq \max_{\mu\in K}\mu|u|\leq p(|u|)=p(u)\quad\text{for all }u\in L.
    \]
	By \eqref{eq.functional1}, for all $\la\in V^\circ$ and $u\in L$,
	\[
	-\la_-|u|\leq \la |u|\leq \la_+ |u| \quad\text{and}\quad \max\big\{\la_+|u|,\la_-|u|\}\leq p(u).
	\]
	This shows that $\la_+,\la_-\in K$ and $\big| \la |u|\big|\leq\max\big\{ \la_+ |u|,\la_-|u|\big\}$ for all $\la\in V^\circ$.\ Hence, by \eqref{eq.seminorm1},
	\[
	p(u)=p(|u|)=\max_{\lambda\in V^\circ} \big|\lambda |u|\big|\leq \sup_{\la\in V^\circ}\max\{\la_+ |u|,\la_-|u|\}\leq \max_{\mu\in K}\mu |u|.
	\]
	We have therefore shown that $p(u)=\max_{\mu\in K}\mu|u|$ for all $u\in L$. On the other hand,
	\[
	p(u)=\max_{\la\in V^\circ}|\la u|\leq \max_{\la\in V^\circ}\big(|\la_+ u|+|\la_- u|\big)\leq 2\max_{\mu\in K}|\mu u|.
	\]
\end{proof}


\end{document}